\documentclass[a4paper,UKenglish,cleveref, autoref, thm-restate]{lipics-v2021}

\usepackage{amssymb,mathtools,tikz,hyperref}
\usetikzlibrary{automata, positioning, arrows}

\newcommand{\N}{\ensuremath{\mathbb{N}} }

\newcommand{\fbins}{\ensuremath{\{0,1\}^*}}
\newcommand{\bin}{\ensuremath{\{0,1\}}}
\newcommand{\infbins}{\ensuremath{\{0,1\}^{\omega}}}
\newcommand{\thh}{\ensuremath{\textrm{th}}}
\renewcommand{\gcd}{\ensuremath{\mathrm{gcd}}}
\newcommand{\finfbins}{\ensuremath{\bin^{\leq \omega}}}
\newcommand{\occb}[2]{\ensuremath{\mathrm{occ}_b(#1,#2)}}
\newcommand{\occ}[2]{\ensuremath{\mathrm{occ}(#1,#2)}}

\newtheorem{question}[theorem]{Question}

\title{Normal Sequences with Non-Maximal Automatic Complexity} 

\titlerunning{Normal Sequences with Non-Maximal Automatic Complexity} 

\author{Liam Jordon\footnote{Corresponding author}}{Department of Computer Science, Maynooth University, Ireland \and \url{https://www.researchgate.net/profile/Liam_Jordon}} {liam.jordon@mu.ie}{https://orcid.org/0000-0003-0583-666X}{Supported by the Irish Research Council's Government of Ireland Postgraduate Scholarship Programme. Grant number: GOIPG/2017/1200}

\author{Philippe Moser}{Department of Computer Science, Maynooth University, Ireland  \and \url{http://www.cs.nuim.ie/~pmoser/}}{philippe.moser@mu.ie}{}{}

\authorrunning{L. Jordon and P. Moser} 

\Copyright{Liam Jordon and Philippe Moser} 



\relatedversion{}
\relatedversiondetails[linktext={ 10.4230/LIPIcs.FSTTCS.2021.47}]{Conference Version}{https://drops.dagstuhl.de/opus/volltexte/2021/15558/}

\nolinenumbers 

\hideLIPIcs  

\EventEditors{Miko{\l}aj Boja\'{n}czyk and Chandra Chekuri}
\EventNoEds{2}
\EventLongTitle{41st IARCS Annual Conference on Foundations of Software Technology and Theoretical Computer Science (FSTTCS 2021)}
\EventShortTitle{FSTTCS 2021}
\EventAcronym{FSTTCS}
\EventYear{2021}
\EventDate{December 15--17, 2021}
\EventLocation{Virtual Conference}
\EventLogo{}
\SeriesVolume{213}
\ArticleNo{7}

\begin{document}

\maketitle

\begin{abstract}
This paper examines Automatic Complexity, a complexity notion introduced by Shallit and Wang in 2001 \cite{DBLP:journals/jalc/ShallitW01}. We demonstrate that there exists a normal sequence $T$ such that $I(T) = 0$ and $S(T) \leq 1/2$, where $I(T)$ and $S(T)$ are the lower and upper automatic complexity rates of $T$ respectively. We furthermore show that there exists a Champernowne sequence $C$, i.e. a sequence formed by concatenating all strings of length one followed by concatenating all strings of length two and so on, such that $S(C) \leq 2/3$.
\end{abstract}

\section{Introduction}

Due to the uncomputability of Kolmogorov complexity, finite-state automata and transducers have acted as a popular setting to study the complexity of finite strings and infinite sequences. In this paper we examine the finite-state based complexity introduced by Shallit and Wang, which is analogous to Sipser's \emph{Distinguishing Complexity} \cite{DBLP:conf/stoc/Sipser83a}, known as \textit{Automatic Complexity} \cite{DBLP:journals/jalc/ShallitW01}. For a string $x$ of length $n$, its automatic complexity $A(x)$ is defined to be the minimum number of states required by any deterministic finite-state automaton such that $x$ is the only string of length $n$ the automaton accepts. A non-deterministic variation was first examined by Hyde \cite{HydeThesis}. In their paper, Shallit and Wang found upper and lower bounds for the automatic complexity of various sets of strings and of prefixes of the infinite Thue-Morse sequence. Expanding on this line of research, Kjos-Hanssen has recently studied the automatic complexity of Fibonacci and Tribonacci sequences \cite{autoFib}.

We continue this line of research by examining the automatic complexity of some \textit{normal} sequences. A binary sequence is normal number in the sense of Borel \cite{borelNormal} if for all $n$, every string of length $n$ occurs as a substring in the sequence with limiting frequency $2^{-n}$. The complexity of normal sequences has been widely studied in the finite-state setting for many years and a review of several old and new results can be found in \cite{DBLP:journals/jcss/KozachinskiyS21}.

Depending on how finite-state complexity is measured, normal sequences may have high or low complexity. For instance, if complexity is defined as compressibility by lossless finite-state compressors, normal sequences have maximum complexity. For example, combining results of Schnorr and Stimm \cite{DBLP:journals/acta/SchnorrS72} and Dai, Lathrop, Lutz and Mayordomo \cite{DBLP:journals/tcs/DaiLLM04} demonstrates that a sequence is normal if and only if it cannot be compressed by any lossless finite-state compressor (see \cite{DBLP:journals/tcs/BecherH13} for a proof). This is also true when the finite-state compressor is equipped with a counter \cite{DBLP:journals/jcss/BecherCH15} or when the reading head is allowed to move in two directions  \cite{DBLP:journals/iandc/CartonH15}. Another definition examines the length of the minimal input required to output a string via finite-state transducers and has been used in \cite{DBLP:journals/tcs/CaludeSR11,DBLP:journals/iandc/CaludeSS16,DotyMoserLossy,DBLP:conf/cie/DotyM07,DBLP:conf/sofsem/JordonM20}. It was demonstrated in Theorem $24$ of \cite{DBLP:journals/iandc/CaludeSS16} that in a complexity based on this approach, one can construct normal sequences with minimal complexity.

As automatic complexity is more of a ``combinatorial'' rather than an ``information content'' measurement\footnote{In an information content measurement, we would like there to be roughly $2^n$ objects with a complexity of $n$ while it trivially holds that all strings of the form $0^n$ and $1^n$ have an automatic complexity of $2$.}, this leads to the question as to how low can the automatic complexity of normal sequences be? Previously the automatic complexity of finite strings produced by linear feedback shift registers which have a maximal number of distinct substrings (otherwise known as $m$-sequences)  \cite{autoLSFR} along with sequences and finite strings which do not contain $k$-powers, i.e. substrings of the form $x^k$, have been studied \cite{DBLP:journals/combinatorics/HydeK15,autoFib,DBLP:journals/jalc/ShallitW01}. Normal sequences by definition contain $x^k$ as a substring infinitely often for every possible pair $(x,k)$. Is there a trade-off between the randomness of normal sequences resulting in high complexity, in the sense that they contain every string as a substring infinitely often, and the fact that some of those substrings have the form $x^k$ which results in low automatic complexity? 

We explore this question by constructing a normal sequence $T$ whose upper automatic complexity rate $S(T)$ is bounded above by $1/2$ and whose lower automatic complexity rate $I(T)$ is $0$. We then study a specific class of normal sequences known as \textit{Champernowne} sequences \cite{champernowne}, i.e. sequences formed by concatenating all strings of length one followed by all strings of length two and so on. It is widely known that Champernowne sequences are incompressible by the Lempel-Ziv $78$ algorithm and results by Lathrop and Strauss show that all sequences incompressible by Lempel-Ziv $78$ are normal \cite{DBLP:conf/sequences/LathropS97}. Due to this restriction on their construction, one may expect Champernowne sequences to have high automatic complexity. However, we demonstrate that there exists a Champernowne sequence $C$ built via a method presented by Pierce and Shields in \cite{pierce2000sequences} that satisfies $S(C) \leq 2/3$. It has previously been seen that Champernowne sequences built via their method are compressible by a variation of the Lempel-Ziv 77 \cite{pierce2000sequences} and the PPM$^*$ \cite{DBLP:conf/sofsem/JordonM21} compression algorithms.

\section{Preliminaries}

 We work with the binary alphabet \bin\, in this paper. A finite \emph{string} is an element of $\fbins$. A \emph{sequence} is an element of $\infbins$. \finfbins denotes the set $\fbins \bigcup \infbins$. The length of a string $x$ is denoted by $|x|$. We say $|S| = \omega$ for $S \in \infbins$. $\lambda$ denotes the string of length $0$. $\bin^n$ denotes the set of strings of length $n$. For $x \in \finfbins$ and $0 \leq i < |x|$, $x[i]$ denotes the $(i+1)^{\thh}$ bit of $x$ with $x[0]$ being the first bit. For $x \in \finfbins$ and $0 \leq i \leq j < |x|$, $x[i..j]$ denotes the \textit{substring} of $x$ consisting of its $(i+1)^{\thh}$ through $(j+1)^{\thh}$ bits. For $x\in \fbins$ and $y\in \finfbins$, $xy$ (sometimes written as $x \cdot y$) denotes the string (or sequence) $x$ concatenated with $y$. For a string $x$, $x^n$ denotes $x$ concatenated with itself $n$  times. For $x \in \finfbins$, a substring $y$ of $x$ is called a $k$-\emph{power} if $y = u^k$ for some string $u$. For $x \in \fbins$ and $y,z \in \finfbins$ such that  $z = xy,$ we call $x$ a \textit{prefix} of $z$ and $y$ a \emph{suffix} of $z$. We write $x[i..]$ to denote the suffix of $x$ beginning with its $(i+1)^{\thh}$ bit. The \emph{lexicographic-length} ordering of $\fbins$ is defined by saying for two strings $x,y$, $x$ comes before $y$ if either $|x| < |y|$ or else $|x| = |y|$ with $x[n] = 0$ and $y[n] = 1$ for the smallest $n$ such that $x[n] \neq y[n]$.
 
 Given strings $x,w$ we use the following notation to count the number of times $w$ occurs as a substring in $x$. The number of occurrences of $w$ as a substring of $x$ is given by
    \begin{align*}
        \occ{w}{x} = |\{i : x[i..i+|w|-1] = w\}|.
        \end{align*}
The block number of occurrences of $w$ as a substring of $x$ is given by
    \begin{align*}\occb{w}{x} = |\{i : x[i .. i+|w|-1] = w \wedge i \equiv 0 \bmod |w| \}|.\end{align*}
For example, $\occ{00}{0000} = 3$ while $\occb{00}{0000}=2$.
 
\noindent Automatic complexity is based on finite automata.
 \begin{definition}
 A \emph{deterministic finite-state automaton (DFA)} is a $4$-tuple $M = (Q,q_0,\delta,F)$, where
\begin{itemize}
    \item $Q$ is a non-empty, finite set of \emph{states},
    \item $q_0 \in Q$ is the \emph{initial state},
    \item $\delta : Q\times \bin \rightarrow Q$ is the \emph{transition function}, 
    \item $F \subseteq Q$ is the set of \emph{final / accepting states}.
\end{itemize}
\label{Auto: Def: DFA}
\end{definition}

\noindent A DFA $M$ can be thought of as a function $M: \fbins \rightarrow Q$ such that for all $x \in \fbins$ and $b \in \bin$, $M$ is defined by the recursion $M(\lambda) = q_0$ and $M(xb) = \delta(M(x),b)$. If $M(x) \in F$, we say $M$ \emph{accepts} $x$. We write $L(M)$ to denote the \emph{language} of $M$, i.e. the set of strings that $M$ accepts.

Shallit and Wang define \emph{automatic complexity} as follows.
\begin{definition}[\cite{DBLP:journals/jalc/ShallitW01}]
Let $x \in \fbins.$ The automatic complexity of $x$, denoted by $A(x)$, is the minimal number of states required by any DFA $M$ such that $L(M) \bigcap \bin^{|x|} = \{x\}.$
\end{definition}

\noindent We say a DFA $M$ \textit{uniquely accepts} a string $x$ if $L(M) \bigcap \bin^{|x|} = \{x\}.$

Shallit and Wang compute the following two ratios to examine the automatic complexity of sequences. \begin{definition}The \emph{lower} and \emph{upper} rates for the automatic complexity of a sequence $T$ are respectively given by
\begin{align*}I(T) =  \liminf\limits_{m \to \infty} \frac{A(T[0..m])}{m+1} \text{ and, } S(T) =  \limsup\limits_{m \to \infty} \frac{A(T[0..m])}{m+1}.\end{align*} \end{definition}
\noindent From the fact that for all $x \in \fbins$ it trivially holds that $A(x) \leq |x| + 2$, it follows that for all $T \in \infbins$, $0 \leq I(T) \leq S(T) \leq 1$.

Normal sequences and \textit{de Bruijn} strings which we use to build \emph{normal} sequences are defined as follows. 
\begin{definition}
A sequence $T\in \infbins$ is normal if for all $x \in \fbins$, \begin{align*}\lim_{m \to \infty} \frac{\occ{x}{T[0..m]}}{m+1} = 2^{-|x|}.\end{align*}
\end{definition}

\begin{definition}[\cite{debruijn1946combinatorial,FlyedeBruijn}]
A \emph{de Bruijn string} of order $n$ is a string $u \in \bin^{2^n}$ such that for all $w \in \bin^n, \occ{w}{u\cdot u[0..n-2]} = 1$.
\end{definition}

For example, $0011$ and $00010111$ are de Bruijn strings of order $2$ and $3$ respectively. We generally use $d_n$ to denote a de Bruijn string of order $n$. It is known that there are $2^{2^{n-1}-n}$ de Bruijn strings of order $n$ unique up to cycling, i.e. the two de Bruijn strings $0011$ and $0110$ are considered the same string for example when counting.

\section{Normal Sequences with Low Automatic Complexity}

In our first result we construct a normal sequence $T$ such that $I(T) = 0,$ that is, infinitely many prefixes have close to minimal automatic complexity. We furthermore show that $S(T) \leq 1/2$, indicating that the sequence does not have high complexity. We require the following variation of a result by Nandakumar and Vangapelli. 

\begin{theorem}[\cite{DBLP:journals/mst/NandakumarV16}]
Let $f:\N \to \N$ be increasing such that for all $n$, $f(n) \geq n^n$. Then every sequence of the form $T = d_1^{f(1)}d_2^{f(2)}d_3^{f(3)} \cdots$ where $d_n$ is a de Bruijn string of order $n$, is normal \footnote{Nandakumar and Vangapelli's original result was for when $f(n) = n^n$. However, their argument easily carries over for $f(n) \geq n^n$ also and this fact has been used by other authors such as in \cite{DBLP:journals/mst/CaludeS18,DBLP:journals/iandc/CaludeSS16}.}. \label{nanda vang thm}
\end{theorem}

\begin{theorem}
There is a normal sequence $T$ such that $I(T) = 0$ and $S(T) \leq \frac{1}{2}.$
\label{Normal low complexity}
\end{theorem}

\begin{proof}
We recursively define the sequence $T = T_1T_2\ldots$ and the function $f:\N \to \N$ as follows. For all $j \geq 1$, let $d_j$ be a de Bruijn string of order $j$ such that if $j$ is odd, $d_j$ begins with a $1$ and if $j$ is even, $d_j$ begins with a $0$. We set $f(1) = 2$ and $T_1 = d_1^{f(1)} = d_{1}^{2}$. For $j \geq 2$, we define $f(j)  = |T_1 \ldots T_{j-1}|^{|T_1 \ldots T_{j-1}|}$ and $T_j = d_j^{f(j)}$. Note that $f(1) > 1$ and for all $j\geq 2$, $f(j) \geq |T_{j-1}|^{|T_{j-1}|}$ and that $|T_{j-1}| = 2^{j-1}f(j-1) \geq j.$ Hence by Theorem \ref{nanda vang thm}, $T$ is normal. For simplicity, we write $\overline{T_j}$ for the prefix $T_1\cdots T_j$ of $T$.


We first show that $I(T) = 0.$ Consider a prefix of the form $\overline{T_n}$. $\overline{T_n}$ is uniquely accepted by the DFA $M_1$ which has a state for each bit of $\overline{T_{n-1}}$ followed by a loop of length $2^{n}$ for the string $d_n$ whose root state is the only accepting state, and an error state. $M_1$ can be seen in Figure \ref{Thm low complexity Machines}. $M_1$ has $|\overline{T_{n-1}}| + 2^n + 1$ states. Thus we have that

\begin{align*}
    \frac{A(\overline{T_{n}})}{|\overline{T_n}|} & \leq \frac{|\overline{T_{n-1}}| + 2^n + 1}{|T_n|+|\overline{T_{n-1}}|} 
     = \frac{|\overline{T_{n-1}}| + 2^n + 1}{2^nf(n) + |\overline{T_{n-1}}|} \\
    & \leq \max\Big\{ \frac{|\overline{T_{n-1}}|}{2^n |\overline{T_{n-1}}|^{|\overline{T_{n-1}}|}} , \frac{2^n + 1}{|\overline{T_{n-1}}|} \Big\} 
     \leq \max\Big\{ \frac{1}{2^n} , \frac{2^n + 1}{(n-1)^{n-1}} \Big\}. \\
\end{align*} 
Hence it follows that $I(T) = 0$.

Next consider an arbitrary prefix $T[0..m]$ of $T$. Let $n$ be largest such that $\overline{T_n}$ is a prefix of $T[0..m]$ but $\overline{T_{n+1}}$ is not. Thus $T[0..m] = \overline{T_n}\cdot w$ for some $w \in \fbins$ and is uniquely accepted by the DFA $M_2$ in Figure \ref{Thm low complexity Machines}. $M_2$ has a state for each bit of $\overline{T_{n-1}}$, followed by a loop of length $2^{n}$ for the string $d_n$, followed by a state for each bit of $w$ and an error state. $M_2$ has $|\overline{T_{n-1}}| + 2^n + |w| + 1$ states.

Consider when $1 \leq |w| \leq 2^n(f(n) - 1) + 2^{n+1}$. Note that \begin{align*}
    |T_{n+1}| = 2^{n+1}f(n+1) =2^{n+1}(f(n+1) - 1) + 2^{n+1} > 2^n(f(n) - 1) + 2^{n+1}
\end{align*}
so $w$ can be this length. Note also that $M_2$ has at most $|\overline{T_n}| + 2^{n+1} + 1$ states for such $w$. Hence for such $w$ we have that \begin{align}
    \frac{A(T[0..m])}{m+1} & \leq \frac{|\overline{T_{n-1}}| + 2^n + |w| + 1}{|\overline{T_n}| + |w| } \notag\\
    & \leq \frac{|\overline{T_{n-1}}| + 2^n + 2^n(f(n) - 1) + 2^{n+1} + 1}{|\overline{T_n} |+ 2^n(f(n) - 1) + 2^{n+1}} \notag\\
    & = \frac{|\overline{T_{n-1}} | + |T_n| + 2^{n+1}+1}{|\overline{T_{n-1}}| + 2|T_n| + 2^n} \notag\\
     &= \frac{|\overline{T_{n-1}}| + 2^n(|\overline{T_{n-1}}|^{|\overline{T_{n-1}}|} + 2) + 1}{|\overline{T_{n-1}}| + 2(2^n|\overline{T_{n-1}}|^{|\overline{T_{n-1}}|}) + 2^n} \notag\\
    & \leq \max \Big \{ \frac{1 + 2^n(|\overline{T_{n-1}}|^{|\overline{T_{n-1}}|-1} + 2|\overline{T_{n-1}}|^{-1})}{1 + 2(2^n |\overline{T_{n-1}}|^{|\overline{T_{n-1}}|-1})}, \frac{1}{2^n} \Big \}. \label{sup 1}
\end{align}
Note Equation \eqref{sup 1} approaches $1/2$ as $n$ increases.

Furthermore consider when $2^n(f(n) - 1) + 2^{n+1} < |w| \leq |T_{n+1}|$. Instead of looping on $d_n$, it becomes more beneficial to loop on $d_{n+1}$ via a DFA similar to $M_1$ in Figure \ref{Thm low complexity Machines} where the accepting state is a single state in the loop depending on the length of $w$. Thus for such prefixes $A(\overline{T_n}\cdot w) \leq |\overline{T_n}| + 2^{n+1} + 1$, i.e. it does not depend on $w$. Hence the ratio $A(T[0..m])/(m+1)$ decreases and approaches $I(T)$ for such $w$.

Therefore, by Equation \eqref{sup 1}, $S(T) \leq \frac{1}{2}$.
\end{proof}
\begin{figure}[h] 
\centering 

\begin{tikzpicture}[->,>=stealth',shorten >=1pt,auto,node distance=2.25cm,
        scale = 1,transform shape]
\node[state,initial] (qs) {};
\node[state, right of = qs, accepting] (q1) {};

\draw (qs) edge[dashed] node[above,midway]{ $\overline{T_{n-1}}$} (q1)

(q1) edge[dashed,loop above] node{$d_n$} (q1);

\end{tikzpicture}
\qquad
\begin{tikzpicture}[->,>=stealth',shorten >=1pt,auto,node distance=2.25cm,
        scale = 1,transform shape]
\node[state,initial] (qs) {};
\node[state, right of = qs] (q1) {};
\node[state, right of = q1, accepting] (q2) {};

\draw (qs) edge[dashed] node[above,midway]{ $\overline{T_{n-1}}$} (q1)

(q1) edge[dashed,loop above] node{$d_n$} (q1)

(q1) edge[dashed] node[above,midway]{ $w$} (q2);

\end{tikzpicture}
\caption{DFA $M_1$ (left) and $M_2$ (right) from Theorem \ref{Normal low complexity}. The error state (the state traversed to if the bit seen is not the expected bit) and arrows to it are not included. By the construction of $T$, $d_n[0] \neq w[0]$ to ensure determinism.}
\label{Thm low complexity Machines}
\end{figure}
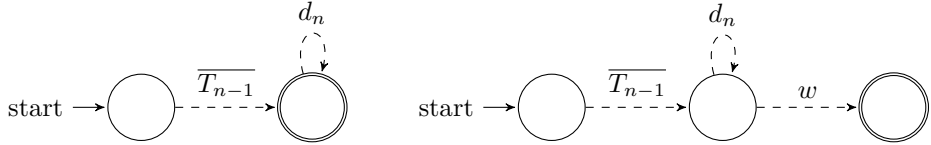

\section{Automatic Complexity of Champernowne Sequences}

In this section we present a Champernowne sequence with an upper automatic complexity rate bounded above by $2/3$. 

\begin{definition}
 A sequence $C$ is a \emph{Champernowne} sequence if $C=C_1C_2C_3\ldots,$ such that for each $n$, $C_n$ is a concatenation of all strings of length $n$ exactly once. That is, for all $x \in \bin^n, \,\occb{x}{C_n} = 1.$
\end{definition}

Unlike Champernowne's original sequence which was a concatenation of all strings in lexicogrpahic-length order ($0100011011000...$), we emphasise that the set of Champernowne sequences do not require strings to be in length-lexicographic order for the construction. There are $2^n!$ possible choices for zone $C_n$ in a Champernowne sequence. For instance, $00011011$ and $11100001$ are two possibilities for $C_2$. 

We now describe Pierce and Shields' construction of Champernowne sequences from \cite{pierce2000sequences}. Suppose we wished to construct substring $C_n$. Let $d_n$ be a de Bruijn string of order $n$. For $0 \leq j \leq 2^n-1$, let $d_{n,j}$ represent a cyclic shift to the left of the first $j$ bits of $d_n.$ That is, $d_{n,j} = d_{n}[j..2^n-1]\cdot d_{n}[0..j-1].$ We write $d_n$ instead of $d_{n,0}$ when no cyclic shift occurs. Note that each $n$ can be written uniquely in the form $n = 2^st$ where $s \geq 0$ and $t \geq 1$ where $t$ is odd. Each substring $C_n$ is broken into further substrings $C_n = B_{n,0}\cdot B_{n,1}\cdots B_{n,2^s - 1}$ where $B_{n,j}$ is a concatenation of $d_{n,j}$ with itself $t$ times. That is, $B_{n,j} = (d_{n,j})^t$. Hence, for example, if $n$ is odd then $C_n = (d_n)^n$ and if $n = 2^k$ for $k \geq 1$, $C_n = d_nd_{n,1}\cdots d_{n,n-1}.$ 

To help the reader visualise this, the result of using the lexicographic least de Bruijn strings of order $3,\,4$ and $6$ to build $C_3,\,C_4$ and $C_6$ via Pierce and Shields' method are provided in Figures \ref{C34} and \ref{C6}. An algorithm to construct the lexicographic least de Bruijn strings was first provided by Martin in 1934 which requires exponential space \cite{martin1934}. Later works by Fredricksen, Kessler and Maiorana led to the FKM-algorithm which only requires $O(n)$ space to construct such strings \cite{DBLP:journals/dm/FredricksenK86,DBLP:journals/dm/FredricksenM78}.

\begin{figure}[h]
\centering
\begin{tabular}{l | l}
000\textcolor{blue}{10111} & 0000\textcolor{blue}{100110101111} \\
\textcolor{blue}{000}10111 & \textcolor{blue}{000}1001101011110\\
00010111 & 0010011010111100\\
 & 0100110101111000
\end{tabular}
\caption{Concatenating the three rows on the left hand side produces the substring $C_3$ and concatenating the four rows on the right hand side produces the substring $C_4$ if $d_3$ and $d_4$ are chosen to be the least lexicographic de Bruijn string of their order respectively.}
\label{C34}
\end{figure}

\begin{figure}[h]
\centering
\begin{tabular}{l}
000000\textcolor{blue}{1000011000101000111001001011001101001111010101110110111111} \\
\textcolor{blue}{0000001000011000101000111001001011001101001111010101110110111111} \\
\textcolor{blue}{0000001000011000101000111001001011001101001111010101110110111111} \\
\textcolor{blue}{00000}10000110001010001110010010110011010011110101011101101111110 \\
0000010000110001010001110010010110011010011110101011101101111110 \\
0000010000110001010001110010010110011010011110101011101101111110
\end{tabular}
\caption{Concatenating the six rows produces the substring $C_6$ where $d_6$ is chosen to be the lexicographic least de Bruijn string of order $6$. The first three rows are $B_0$ while the second three rows are $B_1$.}
\label{C6}
\end{figure}

\noindent In Figures \ref{C34} and \ref{C6} above, the bits shaded in blue indicate the bits of each zone read on a single traversal of the loops described in the proof of Theorem \ref{Champernowne Thm} and shown in Figure \ref{fig: Case 1 2 3 4} in the case where either $n$ or $n+1$ is $3,\,4$ or $6$ respectively.

We re-present Pierce and Shields' proof that their construction builds Champernowne sequences using our notation below. The proof requires some basic results and definitions which can be seen in an undergraduate group theory course. We omit specifics as they are unimportant to the paper as a whole but point towards \cite{rotmanGroups} for those interested.

\begin{lemma}[\cite{pierce2000sequences}]
\label{Enum Proof}
Let $C \in \infbins$ be constructed via Pierce and Shields' construction. Then $C$ is a Champernowne sequence.
\end{lemma}

\begin{proof}
Let $C \in \infbins$ be as described. In order to show that $C$ is a Champernowne sequence we must show that for each zone $C_n$, for all $x \in \bin^n, \, \occb{x}{C_n} = 1.$

Consider substring $C_n$. Let $G_{2^n}$ be the cyclic group of order $2^n$, i.e.  $G_{2^n} = \langle x\, | x^{2^n} = e\rangle$, where $e=x^0$ is the identity element and $x$ is the generator of the group. There exists a bijective mapping  $f:G_{2^n} \rightarrow \{0,1\}^n$ such that for $0 \leq a < 2^n$, $x^a$ is mapped to the substring of $d_n$ of length $n$ beginning in position $a$ when $d_n$ is viewed cyclically. That is, $f(e) = d_n[0.. n-1], \,f(x) = d_n[1 .. n], \ldots \,f(x^{2^n - 1}) = d_n[2^n - 1]\cdot d_n[0 .. n-2].$

Let $s\geq 0$ and $t \geq 1$ where $t$ is odd such that $n=2^st$. Consider the subgroup $\langle x^n\rangle$ of $G_{2^n}$. From group theory it follows that \[|\langle x^n\rangle | = \frac{2^n}{\gcd(n,2^n)} = 2^{2^st - s} = 2^{n - s}.\] So \[\langle x^n\rangle = \bigcup\limits_{i=0}^{2^{n-s}-1}\{x^{in \bmod 2^n}\} = \{e,x^n,x^{2n},\ldots x^{(2^{n-s}-1)n \bmod 2^n}\}.\] Concatenating the result of applying $f$ to each element of $\langle x^n\rangle $ beginning with $e$ in the natural order gives the string \[\sigma = f(e)\cdot f(x^n)\cdot f(x^{2n})\cdots f(x^{(2^{n-s}-1)n \bmod 2^n}).\] $\sigma$ can be thought of as beginning with the prefix of $d_n$ of length $n$, cycling through $d_n$ in blocks of size $n$ until the block containing $d_n$'s suffix of length $n$ is seen. As $(2^{n - s}n)/2^n = t$, we have that $ \sigma = (d_n)^t = B_0$.

As $|G_{2^n}|/|\langle x^n\rangle | = 2^s$, there are $2^s$ cosets of $\langle x^n\rangle$ in $G_{2^n}$. As cosets are disjoint, each represents a different set of $2^{n - s}$ strings of $\bin^n$. Specifically each coset represents some $B_j = (d_{n,j})^t$ block. Therefore, for each $x \in \bin^n$, for some $j \in \{0,\ldots ,2^s-1\}$, $\occb{x}{B_j} = 1$ and $\occb{x}{B_i} = 0$ for each $i\neq j$. Thus $\occb{x}{C_n} = 1$. 
\end{proof}

Before examining the main result of this section, we require the following result from number theory.
\begin{theorem}
For $a,b,c \in \mathbb{Z}$, consider the Diophantine equation $ ax + by = c$. If there exists a solution to the equation $(x_0,y_0)$ where $x_0,y_0 \in \mathbb{Z}$, then all other solutions $(x',y')$ such that $x',y'\in \mathbb{Z}$ are of the form $x' = x_0 + (b/g)d$ and $ y' = y_0 - (a/g)d$  where $d\in \mathbb{Z}$ is arbitrary and $g = \gcd(a,b)$.
\label{dio thm}
\end{theorem}

Henceforth, we write $\mathrm{PSC}$ to denote the set of Champernowne sequences constructed using Pierce and Shields' method such that for each zone $C_n$ of the sequences, a de Bruijn string $d_n$ of order $n$ with prefix $0^n$ was used to construct it.
In the following theorem we show that there exists sequences in $\mathrm{PSC}$ which have non-maximal automatic complexity as their upper automatic complexity rates are bounded above by $2/3$. We also briefly discuss their lower automatic complexity rates in Section \ref{Section lower rate}.
\begin{theorem}
There exists $C \in \mathrm{PSC}$ such that $S(C) \leq \frac{2}{3}$.
\label{Champernowne Thm}
\end{theorem}

\begin{proof}

Let $C \in \mathrm{PSC}$ and consider an arbitrary prefix $C[0..m]$ of $C$. Again we use $\overline{C_n}$ to denote the prefix $C_1C_2\cdots C_n$.  Let $n$ be largest such that $\overline{C_{n+1}}$ is a prefix of $C[0..m]$ but $\overline{C_{n+2}}$ is not.

To examine $A(C[0..m])$ we build automata which make use of two loops which exploit the repetitions of the de Bruijn strings in $C_n$ and $C_{n+1}$. The automata have a single accepting state which depend on the length of the prefix being examined. There are four cases to consider which are dependent on the the value of $n$ and can be seen in Figure \ref{fig: Case 1 2 3 4}. 

\begin{itemize}
    \item \textbf{Case 1}: $n$ is a power of $2$,
    \item \textbf{Case 2}: $n+1$ is a power of $2$,
    \item \textbf{Case 3}: $n$ is even but not a power of $2$,
    \item \textbf{Case 4}: $n+1$ is even but not a power of $2$.
\end{itemize}

Notation wise, we let $v_n$ be the string such that $d_n = 0^n1v_n$. Note that $d_n[n] = 1$ as otherwise the string $0^n$ would appear twice as a substring of $d_n$. Also note that the final bit of $v_n$ must be a $1$.

Suppose we are in Case $3$ where $n = 2^st$ where $s \geq 1$ and $t \geq 3$ where $t$ is odd. We examine Case $3$ as later calculations which maximise the number of states needed require for the possibility that $n+2$ is even but not a power of $2$.

Let $p_{n+2}$ denote the prefix of $C_{n+2}$ such that $C[0..m] = \overline{C_{n+1}}p_{n+2}$. The automaton for Case $3$ in Figure $\ref{fig: Case 1 2 3 4}$ accepts $C[0..m]$ by reading the prefix $\overline{C_{n-1}}\cdot0^n$ state by state, then traversing the first loop $2^s$ times, then reading $0^{2^s + 1}$, then traversing the second loop fully $n$ times and then up to reading $1v_{n+1}$ on the $n+1^{\thh}$ traversal of it. Then, depending on the length of $p_{n+2}$, we read the final $0^{n+1}$ of the second loop and exit it to read the remainder of $C_{n+2}[n+1..]$ as needed. That is, if $|p_{n+2}| \leq n+1$ then the final state is contained in the last $n+1$ states (including the root state) of the second loop of the DFA, else once finishing the loop, we traverse through $|p_{n+2}| - (n+1)$ extra states to the accepting state. 

To see that $C[0..m]$ is accepted by the DFA, note that the traversal through the DFA described above can be factored as $\overline{C_{n-1}}xp_{n+2}$ where \[x= 0^n(1v_nd_n^{t-1}0^{n-1})^{2^s}0^{2^s + 1}(1v_{n+1}0^{n+1})^n(1v_{n+1}).\] Note that $x = C_nC_{n+1}$ since \begin{align*}
    x &= 0^n(1v_nd_n^{t-1}0^{n-1})^{2^s}0^{2^s + 1}(1v_{n+1}0^{n+1})^n(1v_{n+1}) \\
    & = (0^n1v_nd_n^{t-1})0^{n-1}(1v_nd_n^{t-1}0^{n-1})^{2^s -1}0^{2^s+1}(1v_{n+1}0^{n+1})^n(1v_{n+1}) \\
     &= B_0(0^{n-1}1v_nd_n^{t-1}0)0^{n-2}(1v_nd_n^{t-1}0^{n-1})^{2^s -1}0^{2^s+1}(1v_{n+1}0^{n+1})^n(1v_{n+1}) \\
     &= B_0B_1(0^{n-2}1v_nd_n^{t-1}0^2)0^{n-3}(1v_nd_n^{t-1}0^{n-1})^{2^s -2}0^{2^s+1}(1v_{n+1}0^{n+1})^n(1v_{n+1})\\
     & \cdots \tag{Keep repeating this process of sectioning into the $2^s$ blocks}\\
    & = B_0B_1 \cdots B_{2^s - 1}0^{n-2^s}0^{2^s + 1}(1v_{n+1}0^{n+1})^n(1v_{n+1}) \\
    & = C_n (0^{n+1}1v_{n+1})^{n+1} = C_n C_{n+1}.
\end{align*}

Next we show that the DFA uniquely accepts $C[0..m]$. Note that all strings accepted have length \[|\overline{C_{n-1}}| + n + (2^nt - 1)a + 2^s + 1 + 2^{n+1}b  + |p_{n+2}| - (n+1)\] where $a \geq 0$ and $b \geq 1$ if $|p_{n+1}| < (n+1)$, else $b$ can possibly be $0$ too. As stated $(a,b) = (2^s,n+1)$ is a solution to the Diophantine equation \begin{equation}
    |\overline{C_{n-1}}| + n + (2^nt - 1)a + 2^s + 1 + 2^{n+1}b  + |p_{n+2}| - (n+1) = |C[0..m]|. \label{Dio}
\end{equation}%

By Theorem \ref{dio thm}, as the first loop has odd length and the second has even length, all solutions to \eqref{Dio} take the form $(2^s + 2^{n+1}c, \,n+1 - (2^nt - 1)c)$ where $c \in \mathbb{Z}$. As $2^s = n/t$ and $t \geq 3$, we have that $(n+1) - (2^nt - 1)c < 0$ when $c > 0$ and $2^s + 2^{n+1}c < 0$ when $c < 0$, it follows that $c=0$ is the only possibility that gives non-negative integer solutions, i.e. $C[0..m]$ is uniquely accepted by the DFA.

The number of states of the automaton is bounded above by \[|\overline{C_{n-1}}| + n + 2^nt + 2^s + 2^{n+1} + |p_{n+2}|.\] As $2^s \leq n/3$, we then have that \begin{align}
    \frac{A(C[0..m])}{m+1} & \leq \frac{|\overline{C_{n-1}}| + 2^nt + \frac{4n}{3} + 2^{n+1} + |p_{n+2}|}{|\overline{C_{n+1}}| + |p_{n+2}|}. \label{ineq}
\end{align}


\noindent Hence for $|p_{n+2}| \leq 2^n(n-t) - \frac{n}{3} + 1 + 2^{n+2}(\frac{n+2}{2})$ we find that

\begin{align}
    \frac{A(C[0..m])}{m+1} & \leq \frac{|\overline{C_{n-1}}| + 2^nt + \frac{4n}{3} + 2^{n+1} + |p_{n+2}|}{|\overline{C_{n+1}}| + |p_{n+2}|} \notag\\
    & \leq \frac{|\overline{C_{n-1}}| + |C_n| + n + 2^{n+1} + 1 + 2^{n+2}(\frac{n+2}{2})}{|\overline{C_{n+1}}| - \frac{n}{3} + 1 + 2^{n+2}(\frac{n+2}{2})} \notag\\
    & = \frac{|\overline{C_n}| + n + 2^{n+1} + 1 + 2^{n+2}(\frac{n+2}{2})}{|\overline{C_{n+1}}| - \frac{n}{3} + 1 + 2^{n+2}(\frac{n+2}{2})} .\label{case 3 auto proof}
\end{align} We note that taking the limit of \eqref{case 3 auto proof} as $n$ increases has a value of $2/3.$

However as $|p_{n+2}|$ increases, it becomes more advantageous to use a loop for the repetitions in $C_{n+2}$ as opposed to the loop for $C_n$ (similar to the proof of Theorem \ref{Normal low complexity}). Worst case scenario, $n+2$ has the form $2^{s'}t'$ for $s'\geq 1$ and $t' \geq 3$ where $t'$ is odd. This results in an automaton of Case $4$ in Figure \ref{fig: Case 1 2 3 4} where the accepting state is one of that states of the second loop. One can show the prefix is uniquely accepted as before with a similar argument. Such an automaton requires at most $|\overline{C_n}| + 2^{n+1} + n + 1 + 2^{n+2}(\frac{n+2}{2})$ states (as $t' \leq (n+2)/2$).

Hence for $j \geq 1$ such that $2^n(n-t) - \frac{n}{3} + 1 + 2^{n+2}(\frac{n+2}{2}) \leq |p_{n+2}|+j < |C_{n+2}|$ we use the automaton from Case $4$ and get that \begin{align}
    \frac{A(C[0..m])}{m+1} & \leq \frac{|\overline{C_n}| + n + 2^{n+1} + 1 + 2^{n+2}(\frac{n+2}{2})}{|\overline{C_{n+1}}| - \frac{n}{3} + 1 + 2^{n+2}(\frac{n+2}{2}) + j}. \label{case 4 constantt state}
\end{align} 
One can see that for such $p_{n+2}$, the number of states of the automaton used to calculate \eqref{case 4 constantt state} remains constant and the ratio decreases as $j$ increases.

Similar calculations for the other three cases show that none achieve an upper bound greater than $2/3$ (see appendix). Therefore $S(C) \leq 2/3$.
\end{proof}


\subsection{Discussion on Lower Bounds for Champernowne Sequences}
\label{Section lower rate}
Let $C \in \mathrm{PSC}$ satisfy Theorem \ref{Champernowne Thm}. As part of Theorem \ref{Champernowne Thm}, we do not provide any insight into the value of $I(C)$. Currently many of the techniques for calculating lower bounds rely on the absence of $k$-powers (such as proofs in \cite{autoFib,DBLP:journals/jalc/ShallitW01}). In particular, Shallit and Wang show that for every $x$ without $k$-powers, $x$ satisfies $A(x) \geq (|x|+1)/k.$ However, as $C$ is a Champernowne sequence, long enough prefixes contain $k$-powers, i.e. there eventually is a substring $x$ such that $x = u^k$ for some string $u$.

However, one can easily identify an upper bound for $I(C)$ as follows. Consider prefixes of the form $\overline{C_{n+1}}$ where $n$ is a power of $2$, i.e. we are in Case $1$. The automaton in Figure \ref{fig: Case 1 2 3 4} for Case $1$ where the final state is contained appropriately in the second loop will uniquely accept the prefix and simple calculations give us that $I(C) \leq 1/4.$

One prefix $x$ of $C$ such that $A(x)/|x| < 1/4$ which is in Case $1$ is $\overline{C_{65}}$. The automaton for Case 1 in Figure \ref{fig: Case 1 2 3 4} which uniquely accepts $\overline{C_{65}}$ has $n_1 = |\overline{C_{63}}| + 2^{64} + 2^{65} +128$ 
states. However, the number of states can be reduced further by using two more loops for zones $C_{62}$ and $C_{63}$ instead of having a state for each of their bits. Consider the DFA $\widehat{M}$ shown in Figure \ref{fig: C 65}:

$\widehat{M}$ reads the prefix $\overline{C_{61}} \cdot 0^{62}$. It then traverses a loop for $1v_{62}(d_{62})^{30}0^{61}$. It then reads $0^3$ and enters a loop for the string $(1v_{63}0^{63})^{21}$. Following this it reads $01v_{64}0^{63}$ and then enters a loop for the string $(1v_{64}0^{63})^{7}$. It then reads $0^{65}$ and enters a loop for the string $(1v_{65}0^{65})^{5}$, with the final state being that state of this loop after reading $(1v_{65}0^{65})^4\cdot1v_{65}.$ $\widehat{M}$ can be thought of as combining the DFAs from Figure \ref{fig: Case 1 2 3 4}, but altering the length of the loops for the zones. Strings of length $|\overline{C_{65}}|$ that $\widehat{M}$ accepts satisfy the equation
\begin{equation}
    |\overline{C_{61}}| + (2^{62}\cdot 31 - 1)a + (21\cdot2^{63})b + (7\cdot(2^{64} -1))c  + (5\cdot2^{65})d + 65 + 2^{64} = |\overline{C_{65}}| \label{65 eqn}
\end{equation}
where it must hold that $d \geq 1$.

$a=2$, $b = 3$, $c = 9$ and $d = 13$ is the only non-negative integer solution to Equation \eqref{65 eqn} and so $\widehat{M}$ uniquely accepts $C_{65}$. $\widehat{M}$ has $n_2 = |\overline{C_{61}}| + 31\cdot 2^{62} + 7\cdot 2^{63} + 8\cdot 2^{64} + 5\cdot 2^{65} + 120$ 
states which is less than $n_1$. Hence $\widehat{M}$ gives us that $A(\overline{C_{65}})/|\overline{C_{65}}| < 0.173 < 1/4.$

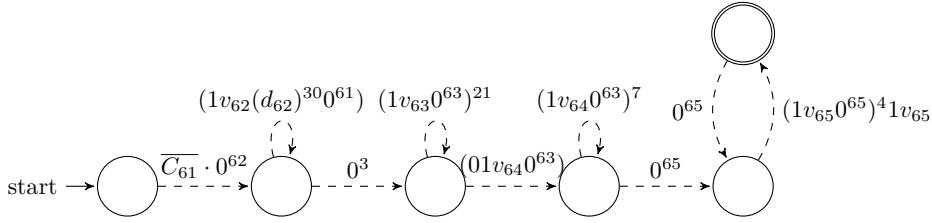
\begin{figure}[ht] 
\centering 
\begin{tikzpicture}[->,>=stealth',shorten >=1pt,auto,node distance=2.25cm,
        scale = 0.9,transform shape]
\node[state,initial] (qs) {};
\node[state, right of = qs] (q1){};
\node[state, right of = q1] (q2){};
\node[state, right of = q2] (q3){};
\node[state, right of = q3] (q4){};
\node[state, accepting, above of = q4] (q5){};

\draw
(qs) edge[dashed] node[above]{$\overline{C_{61}}\cdot 0^{62}$} (q1)

(q1) edge[dashed, loop above] node{$(1v_{62}(d_{62})^{30}0^{61})$} (q1)

(q1) edge[dashed, above] node[above]{$0^3$} (q2)

(q2) edge[dashed, loop above] node{$(1v_{63}0^{63})^{21}$} (q2)

(q2) edge[dashed, above] node[above]{$(01v_{64}0^{63})$} (q3)

(q3) edge[dashed, loop above] node{$(1v_{64}0^{63})^7$} (q3)

(q3) edge[dashed, above] node[above]{$0^{65}$} (q4)

(q4) edge[dashed, bend right] node[right]{$(1v_{65}0^{65})^4 1v_{65}$} (q5)

(q5) edge[dashed, bend right] node[left]{$0^{65}$} (q4);

\end{tikzpicture}

\caption{Automaton $\widehat{M}$ for $\overline{C_{65}}$. The dashed arrows represent the missing states belonging to their labels. The
\label{fig: C 65} error state (the state traversed to if the bit seen is not the expected bit) and arrows to it are not included.}
\end{figure}

While the above demonstrates that more than two loops can be used, the size of the loops are limited in each case. The following proposition demonstrates that if reading a zone $C_j$ where $j$ is odd, any loop traversed used to read a substring $x$ of $C_j$, if $|x| \geq j$ then $|x|$ must be a multiple of $2^j$.

\begin{proposition}
Let $C \in \mathrm{PSC}$. Let $j$ be odd and $C'$ be a prefix of $C$ containing the substring $C_j$. Let $p_0,p_1,\ldots p_{2^jj}$ be the sequence of states an automaton which uniquely accepts $C'$ traverses while reading $C_j$. If there is some subsequence of states $p_i,\ldots p_{i+l},\ldots  p_{i + 2l}$ where $l \geq j$ and for all $0 \leq m \leq l$, $p_{i+m} = p_{i + l + m}$, then $l$ must be a multiple of $2^j$.
\label{loop lemma}
\end{proposition}

\begin{proof}

Let $C$, $C'$ and $j$ be as above. First suppose a loop of length $l$ is traversed where $j \leq l \leq 2^j - 1$. Let $x$ be the substring of $C_j$ read during the loop. Hence $x = yz$ where $|y| = j$ and $|z| \leq 2^j - j - 1$. Suppose the loop is traversed twice in a row indicating that $x^2 = yzyz$ is a substring of $C_j$. Consider $yzy$ which has length at most $2^j + j -1$. By the construction of $C_j$, $yzy$ is a prefix of $d_{j,k} \cdot d_{j,k}[0..j-2]$ for some $k$. By the nature of de Bruijn stings, $d_{j,k} \cdot d_{j,k}[0..j-2]$ contains every string of length $j$ as a substring exactly once. However, $y$ is a substring of length $j$ contained twice giving us a contradiction.

Next suppose a loop of length $l$ is traversed  where $d\cdot2^j < l < (d+1)\cdot 2^{j}$ for some $d \leq \lfloor j/2 \rfloor$ as if $d > \lfloor j/2 \rfloor$, traversing the loop twice would result in a string longer than $C_n$ being read. Let $x$ be the string read while traversing the loop. Hence $x = yz$ where $|y| = d\cdot 2^j$ and $1 \leq |z| < 2^j$. Suppose the loop is traversed twice in a row indicating that $x^2 = yzyz$ is a substring of $C_j$. By the construction of $C_j$, $yzy = (d_{j,k})^d z (d_{j,k})^d$ for some $k$ where $z$ is a prefix of $d_{j,k}$. This forces $z = \lambda$ or $z = d_{j,k}$ which is a contradiction.
\end{proof}

Similar results to the above proposition can be shown for $n$ even also. For instance, if $n$ is a power of $2$, loops of length larger than $n$ in zone $C_n$ have to be a multiple of $2^n-1$. Details can be found in the appendix.

\begin{question}
    Let $C \in \mathrm{PSC}$ satisfy Theorem \ref{Champernowne Thm}. Is there a limit to the number of beneficial loops we can use to ensure prefixes of $C$ are uniquely accepted? Finding the value of $I(C)$ is left as an open question. For instance, is $I(C) > 0$?
\end{question}

\begin{figure}[ht] 
\centering 
\begin{tikzpicture}[->,>=stealth',shorten >=1pt,auto,node distance=2.25cm,
        scale = 0.9,transform shape]
\node[state,initial] (qs) {$1$};
\node[state, below of=qs]  (A0)  {};
\node[state, right of=A0]  (A1)  {};
\node[state,  above of =A1]  (Au)  {};
\node[state, below of = A0] (B0) { };
\node[state, right of=B0]  (B1)  {};
\node[state,  below of =B1]  (Bu) {};
\node[state, below of = B0,accepting] (f) {};

\draw(qs) edge[dashed] node[left,midway]{ $\overline{C_{n-1}}0^n$} (A0)

(A0)  edge[right] node[below,midway]{ $1$} (A1)

(A1) edge[dashed] node[right,midway]{$v_n$} (Au)

(Au) edge[dashed] node[right,midway]{$0^{n-1}$} (A0)


(A0) edge[dashed] node[left,midway]{$0^{n+1}$} (B0)

(B0)  edge[right] node[above,midway]{ $1$} (B1)

(B1) edge[dashed] node[right,midway]{$v_{n+1}$} (Bu)

(Bu) edge[dashed] node[right,midway]{$0^{n+1}$} (B0)

(B0) edge[dashed] node[left,midway]{$C_{n+2}[n+1..]$} (f);

\node[state,initial, below of = f] (q3) {$3$};
\node[state, below of=q3] (c0)  {};
\node[state, below of=c0]  (C0)  {};
\node[state, right of=C0]  (C1)  {};
\node[state,  above of =C1]  (Cu)  {};
\node[state,  right of =q3]  (Cx)  {};
\node[state, below of = C0] (D0) { };
\node[state, right of=D0]  (D1)  {};
\node[state,  below of =D1]  (Du) {};
\node[state, below of = D0,accepting] (f2) {};

\node [right=of C0,  xshift=-20mm, yshift=14mm] {$0^{n-1}$}; 

\draw(q3) edge[dashed] node[left,midway]{ $\overline{C_{n-1}}$} (c0)

(c0) edge[dashed] node[left,midway]{ $0^{n}$} (C0)

(C0)  edge[right] node[below,midway]{ $1$} (C1)

(C1) edge[dashed] node[right,midway]{$v_n$} (Cu)

(Cu) edge[dashed] node[right,midway]{$d_n^{t-1}$} (Cx)

(Cx) edge[dashed] (C0)

(C0) edge[dashed] node[left,midway]{$0^{2^s+1}$} (D0)

(D0)  edge[right] node[below,midway]{ $1$} (D1)

(D1) edge[dashed] node[right,midway]{$v_{n+1}$} (Du)

(Du) edge[dashed] node[right,midway]{$0^{n+1}$} (D0)

(D0) edge[dashed] node[left,midway]{$C_{n+2}[n+1..]$} (f2);
\end{tikzpicture}
\qquad
\begin{tikzpicture}[->,>=stealth',shorten >=1pt,auto,node distance=2.25cm,
        scale = 0.9,transform shape]
\node[state,initial] (qs) {2};
\node[state, below of=qs] (A0) {};
\node[state, right of=A0]  (A1)  {};
\node[state,  above of =A1]  (Au)  {};
\node[state, below of = A0] (B0) { };
\node[state, right of=B0]  (B1)  {};
\node[state,  below of =B1]  (Bu) {};
\node[state, below of = B0,accepting] (f) {};

\draw(qs) edge[dashed] node[left,midway]{ $\overline{C_{n-1}}0^n$} (A0)

(A0)  edge[right] node[below,midway]{ $1$} (A1)

(A1) edge[dashed] node[right,midway]{$v_n$} (Au)

(Au) edge[dashed] node[right,midway]{$0^n$} (A0)

(A0) edge node[left,midway]{$0$} (B0)

(B0)  edge[right] node[above,midway]{ $1$} (B1)

(B1) edge[dashed] node[right,midway]{$v_{n+1}$} (Bu)

(Bu) edge[dashed] node[right,midway]{$0^n$} (B0)

(B0) edge[dashed] node[left,midway]{$C_{n+2}$} (f);

\node[state,initial, below of = f] (q4) {$4$};
\node[state, below of=q4]  (C0)  {};
\node[state, right of=C0]  (C1)  {};
\node[state,  above of =C1]  (Cu)  {};
\node[state, below of = C0] (D0) {};
\node[state, right of=D0]  (D1)  {};
\node[state,  below of =D1]  (Du) {};
\node[state,  below of =Du]  (Dx) {};

\node[state, left of = Dx,accepting] (f2) {};

\node [right=of D0,  xshift=-20mm, yshift=-14mm] {$0^{n}$}; 

\draw(q4) edge[dashed] node[left,midway]{ $\overline{C_{n-1}}0^n$} (C0)

(C0)  edge[right] node[below,midway]{ $1$} (C1)

(C1) edge[dashed] node[right,midway]{$v_n$} (Cu)

(Cu) edge[dashed] node[right,midway]{$0^n$} (C0)

(C0) edge node[left,midway]{$0$} (D0)

(D0)  edge[right] node[above,midway]{ $1$} (D1)

(D1) edge[dashed] node[right,midway]{$v_{n+1}$} (Du)

(Du) edge[dashed] node[right,midway]{$d_{n+1}^{t'-1}$} (Dx)

(Dx) edge[dashed] (D0)

(D0) edge[dashed] node[left,midway]{$0^{2^{s'} + 1}C_{n+2}[n+2..]$} (f2);

\end{tikzpicture}

\caption{Automaton for Case $1$ (top left), Case $2$ (top right),  Case $3$ (bottom left) and Case $4$ (bottom right). The dashed arrows represent the missing states belonging to their labels. The error state (the state traversed to if the bit seen is not the expected bit) and arrows to it are not included in each of the four diagrams. We point the reader to Figures \ref{C34} and \ref{C6} to help visualise the bits read on a single traversal of a loop.}
\label{fig: Case 1 2 3 4}
\end{figure}
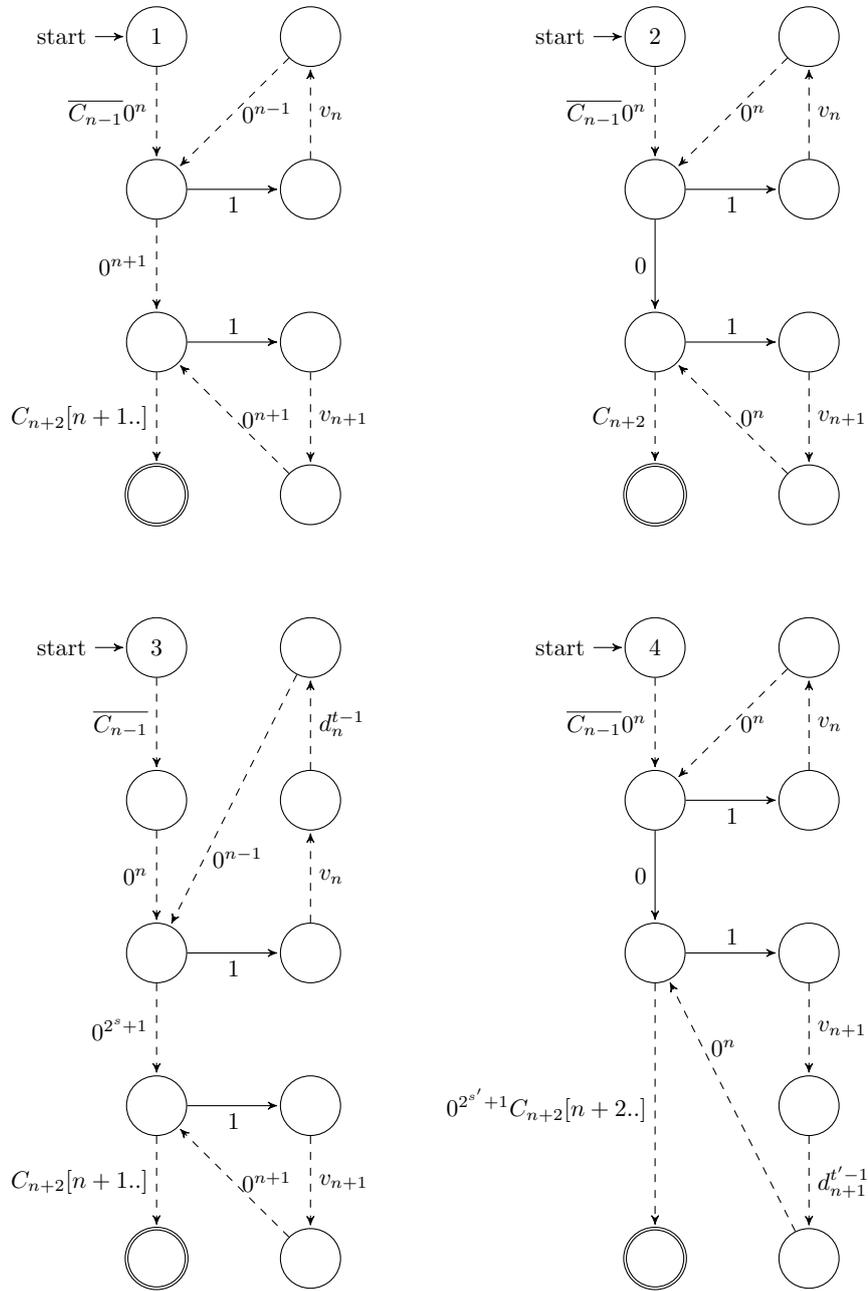


\clearpage
\bibliography{biblio.bib}

\begin{thebibliography}{10}

\bibitem{DBLP:journals/jcss/BecherCH15}
Ver{\'{o}}nica Becher, Olivier Carton, and Pablo~Ariel Heiber.
\newblock Normality and automata.
\newblock {\em J. Comput. Syst. Sci.}, 81(8):1592--1613, 2015.
\newblock \href {https://doi.org/10.1016/j.jcss.2015.04.007}
  {\path{doi:10.1016/j.jcss.2015.04.007}}.

\bibitem{DBLP:journals/tcs/BecherH13}
Ver{\'{o}}nica Becher and Pablo~Ariel Heiber.
\newblock Normal numbers and finite automata.
\newblock {\em Theor. Comput. Sci.}, 477:109--116, 2013.
\newblock \href {https://doi.org/10.1016/j.tcs.2013.01.019}
  {\path{doi:10.1016/j.tcs.2013.01.019}}.

\bibitem{borelNormal}
{\'{E}}mile Borel.
\newblock Les probabilités dénombrables et leurs applications arithmétiques.
\newblock {\em Rendiconti del Circolo Matematico di Palermo}, 27(1):247--271,
  1909.
\newblock \href {https://doi.org/10.1007/BF03019651}
  {\path{doi:10.1007/BF03019651}}.

\bibitem{debruijn1946combinatorial}
Nicolaas Govert~De Bruijn.
\newblock A combinatorial problem.
\newblock In {\em Proc. Koninklijke Nederlandse Academie van Wetenschappen},
  volume~49, pages 758--764, 1946.

\bibitem{DBLP:journals/tcs/CaludeSR11}
Cristian~S. Calude, Kai Salomaa, and Tania Roblot.
\newblock Finite state complexity.
\newblock {\em Theor. Comput. Sci.}, 412(41):5668--5677, 2011.
\newblock \href {https://doi.org/10.1016/j.tcs.2011.06.021}
  {\path{doi:10.1016/j.tcs.2011.06.021}}.

\bibitem{DBLP:journals/mst/CaludeS18}
Cristian~S. Calude and Ludwig Staiger.
\newblock Liouville, computable, {B}orel normal and {M}artin-l{\"{o}}f random
  numbers.
\newblock {\em Theory Comput. Syst.}, 62(7):1573--1585, 2018.
\newblock \href {https://doi.org/10.1007/s00224-017-9767-8}
  {\path{doi:10.1007/s00224-017-9767-8}}.

\bibitem{DBLP:journals/iandc/CaludeSS16}
Cristian~S. Calude, Ludwig Staiger, and Frank Stephan.
\newblock Finite state incompressible infinite sequences.
\newblock {\em Inf. Comput.}, 247:23--36, 2016.
\newblock \href {https://doi.org/10.1016/j.ic.2015.11.003}
  {\path{doi:10.1016/j.ic.2015.11.003}}.

\bibitem{DBLP:journals/iandc/CartonH15}
Olivier Carton and Pablo~Ariel Heiber.
\newblock Normality and two-way automata.
\newblock {\em Inf. Comput.}, 241:264--276, 2015.
\newblock \href {https://doi.org/10.1016/j.ic.2015.02.001}
  {\path{doi:10.1016/j.ic.2015.02.001}}.

\bibitem{champernowne}
D.~G. Champernowne.
\newblock The construction of decimals normal in the scale of ten.
\newblock {\em J. Lond. Math. Soc.}, s1-8(4):254--260, 1933.
\newblock \href {https://doi.org/10.1112/jlms/s1-8.4.254}
  {\path{doi:10.1112/jlms/s1-8.4.254}}.

\bibitem{DBLP:journals/tcs/DaiLLM04}
Jack~J. Dai, James~I. Lathrop, Jack~H. Lutz, and Elvira Mayordomo.
\newblock Finite-state dimension.
\newblock {\em Theor. Comput. Sci.}, 310(1-3):1--33, 2004.
\newblock \href {https://doi.org/10.1016/S0304-3975(03)00244-5}
  {\path{doi:10.1016/S0304-3975(03)00244-5}}.

\bibitem{DotyMoserLossy}
David Doty and Philippe Moser.
\newblock Finite-state dimension and lossy decompressors.
\newblock {\em CoRR}, abs/cs/0609096, 2006.
\newblock \href {http://arxiv.org/abs/cs/0609096} {\path{arXiv:cs/0609096}}.

\bibitem{DBLP:conf/cie/DotyM07}
David Doty and Philippe Moser.
\newblock Feasible depth.
\newblock In {\em CiE 2007: Computation and Logic in the Real World - Siena,
  Italy}, volume 4497 of {\em {LNCS}}, pages 228--237. Springer, 2007.
\newblock \href {https://doi.org/10.1007/978-3-540-73001-9\_24}
  {\path{doi:10.1007/978-3-540-73001-9\_24}}.

\bibitem{DBLP:journals/dm/FredricksenK86}
Harold Fredricksen and Irving~J. Kessler.
\newblock An algorithm for generating necklaces of beads in two colors.
\newblock {\em Discret. Math.}, 61(2-3):181--188, 1986.
\newblock \href {https://doi.org/10.1016/0012-365X(86)90089-0}
  {\path{doi:10.1016/0012-365X(86)90089-0}}.

\bibitem{DBLP:journals/dm/FredricksenM78}
Harold Fredricksen and James Maiorana.
\newblock Necklaces of beads in {$k$} colors and {$k$}-ary de {B}ruijn
  sequences.
\newblock {\em Discret. Math.}, 23(3):207--210, 1978.
\newblock \href {https://doi.org/10.1016/0012-365X(78)90002-X}
  {\path{doi:10.1016/0012-365X(78)90002-X}}.

\bibitem{HydeThesis}
Kayleigh Hyde.
\newblock Nondeterministic finite state complexity.
\newblock Master's thesis, University of Hawaii at Manoa, 2013.
\newblock Accessed: April 20, 2021.
\newblock URL: \url{http://hdl.handle.net/10125/29507}.

\bibitem{DBLP:journals/combinatorics/HydeK15}
Kayleigh Hyde and Bj{\o}rn Kjos{-}Hanssen.
\newblock Nondeterministic automatic complexity of overlap-free and almost
  square-free words.
\newblock {\em Electron. J. Comb.}, 22(3):P3.22, 2015.
\newblock \href {https://doi.org/10.37236/4851} {\path{doi:10.37236/4851}}.

\bibitem{DBLP:conf/sofsem/JordonM20}
Liam Jordon and Philippe Moser.
\newblock On the difference between finite-state and pushdown depth.
\newblock In {\em 46th International Conference on Current Trends in Theory and
  Practice of Informatics, {SOFSEM} 2020, Limassol, Cyprus}, volume 12011 of
  {\em {LNCS}}, pages 187--198. Springer, 2020.
\newblock \href {https://doi.org/10.1007/978-3-030-38919-2\_16}
  {\path{doi:10.1007/978-3-030-38919-2\_16}}.

\bibitem{DBLP:conf/sofsem/JordonM21}
Liam Jordon and Philippe Moser.
\newblock {A normal sequence compressed by {P}{P}{M}* but not by {L}empel-{Z}iv
  78}.
\newblock In {\em 47th International Conference on Current Trends in Theory and
  Practice of Computer Science, {SOFSEM} 2021, Bolzano-Bozen, Italy}, volume
  12607 of {\em {LNCS}}, pages 389--399. Springer, 2021.
\newblock \href {https://doi.org/10.1007/978-3-030-67731-2\_28}
  {\path{doi:10.1007/978-3-030-67731-2\_28}}.

\bibitem{autoFib}
Bj{\o}rn Kjos{-}Hanssen.
\newblock Automatic complexity of {F}ibonacci and {T}ribonacci words.
\newblock {\em Discrete Applied Mathematics}, 289:446 -- 454, 2021.
\newblock \href {https://doi.org/10.1016/j.dam.2020.10.014}
  {\path{doi:10.1016/j.dam.2020.10.014}}.

\bibitem{autoLSFR}
Bjørn Kjos-Hanssen.
\newblock Automatic complexity of shift register sequences.
\newblock {\em Discrete Mathematics}, 341(9):2409--2417, 2018.
\newblock \href {https://doi.org/10.1016/j.disc.2018.05.015}
  {\path{doi:10.1016/j.disc.2018.05.015}}.

\bibitem{DBLP:journals/jcss/KozachinskiyS21}
Alexander Kozachinskiy and Alexander Shen.
\newblock Automatic {K}olmogorov complexity, normality, and finite-state
  dimension revisited.
\newblock {\em J. Comput. Syst. Sci.}, 118:75--107, 2021.
\newblock \href {https://doi.org/10.1016/j.jcss.2020.12.003}
  {\path{doi:10.1016/j.jcss.2020.12.003}}.

\bibitem{DBLP:conf/sequences/LathropS97}
James~I. Lathrop and Martin Strauss.
\newblock A universal upper bound on the performance of the {L}empel-{Z}iv
  algorithm on maliciously-constructed data.
\newblock In {\em Compression and Complexity of {SEQUENCES} 1997, Positano,
  Amalfitan Coast, Salerno, Italy, June 11-13, 1997, Proceedings}, pages
  123--135. {IEEE}, 1997.
\newblock \href {https://doi.org/10.1109/SEQUEN.1997.666909}
  {\path{doi:10.1109/SEQUEN.1997.666909}}.

\bibitem{martin1934}
M.~H. Martin.
\newblock A problem in arrangements.
\newblock {\em Bull. Amer. Math. Soc.}, 40(12):859--864, 1934.
\newblock \href {https://doi.org/10.1090/S0002-9904-1934-05988-3}
  {\path{doi:10.1090/S0002-9904-1934-05988-3}}.

\bibitem{DBLP:journals/mst/NandakumarV16}
Satyadev Nandakumar and Santhosh~Kumar Vangapelli.
\newblock Normality and finite-state dimension of {L}iouville numbers.
\newblock {\em Theory Comput. Syst.}, 58(3):392--402, 2016.
\newblock \href {https://doi.org/10.1007/s00224-014-9554-8}
  {\path{doi:10.1007/s00224-014-9554-8}}.

\bibitem{pierce2000sequences}
Larry~A. Pierce and Paul~C. Shields.
\newblock Sequences incompressible by {S}{L}{Z} ({L}{Z}{W}), yet fully
  compressible by {U}{L}{Z}.
\newblock In {\em Numbers, Information and Complexity}, pages 385--390.
  Springer, 2000.
\newblock \href {https://doi.org/10.1007/978-1-4757-6048-4\_32}
  {\path{doi:10.1007/978-1-4757-6048-4\_32}}.

\bibitem{rotmanGroups}
Joseph~J. Rotman.
\newblock {\em An Introduction to the Theory of Groups}.
\newblock Graduate Texts in Mathematics. Springer, New York, 1995.
\newblock ISBN: 978-1-4612-8686-8.
\newblock \href {https://doi.org/10.1007/978-1-4612-4176-8}
  {\path{doi:10.1007/978-1-4612-4176-8}}.

\bibitem{FlyedeBruijn}
Camille~Flye Sainte-Marie.
\newblock Solution to question nr. {$48$}.
\newblock In {\em L'Interm\'ediaire des Math\'ematiciens}, volume~1, pages
  107--110, 1894.

\bibitem{DBLP:journals/acta/SchnorrS72}
Claus{-}Peter Schnorr and H.~Stimm.
\newblock Endliche {A}utomaten und {Z}ufallsfolgen.
\newblock {\em Acta Inf.}, 1:345--359, 1972.
\newblock \href {https://doi.org/10.1007/BF00289514}
  {\path{doi:10.1007/BF00289514}}.

\bibitem{DBLP:journals/jalc/ShallitW01}
Jeffrey~O. Shallit and Ming{-}wei Wang.
\newblock Automatic complexity of strings.
\newblock {\em J. Autom. Lang. Comb.}, 6(4):537--554, 2001.
\newblock \href {https://doi.org/10.25596/jalc-2001-537}
  {\path{doi:10.25596/jalc-2001-537}}.

\bibitem{DBLP:conf/stoc/Sipser83a}
Michael Sipser.
\newblock A complexity theoretic approach to randomness.
\newblock In {\em Proceedings of the 15th Annual {ACM} Symposium on Theory of
  Computing, 1983, Boston, Massachusetts, {USA}}, pages 330--335. {ACM}, 1983.
\newblock \href {https://doi.org/10.1145/800061.808762}
  {\path{doi:10.1145/800061.808762}}.

\end{thebibliography}

\newpage
\section*{Appendix}

\subsection*{The other three cases for the proof of Theorem \ref{Champernowne Thm}}
The following is the reasoning for why Equation \eqref{case 3 auto proof} from the proof of Theorem \ref{Champernowne Thm} gives us the upper bound of $S(C)$. We perform similar calculations as in the proof for Cases $1,2$ and $4$ to see this.

\textbf{Case 1:} $n$ is a power of $2$.

Suppose we are in Case $1$. Let $p_{n+2}$ be such that $C[0..m] = \overline{C_{n+1}}p_{n+2}$. We have a automaton as in Case $1$. It requires at most $|\overline{C_{n-1}}|+ 1 + 2n + 2^n + 2^{n+1} + |p_{n+2}|$ states. Thus \begin{align*}
    \frac{A(C[0..m])}{m+1} & \leq \frac{|\overline{C_{n-1}}| + 1 + 2n + 2^n + 2^{n+1} + |p_{n+2}|}{|\overline{C_{n+1}}| + |p_{n+2}|}.
\end{align*}

For $|p_{n+2}|$ long enough, it is more beneficial to loop in $C_{n+2}$ instead of $C_n$ and use an automaton in the style of Case $4$ since worst case scenario $n+2$ has the form $2^{s'}t'$. Such an automaton has at most $|\overline{C_n}| + 1 + n + 2^{n+1} + 2^{n+2}(\frac{n+2}{2})$ states.  

So for $|p_{n+2}| \leq 2^n(n-1) -n + 2^{n+2}(\frac{n+2}{2})$ \begin{align*}
    \frac{A(C[0..m])}{m+1} & \leq \frac{|\overline{C_{n-1}}| + 1 + 2n + 2^n + 2^{n+1} + |p_{n+2}|}{|\overline{C_{n+1}}| + |p_{n+2}|} \\
    &\leq \frac{|\overline{C_{n-1}}| + 1 + 2n - n + 2^n + 2^n(n-1) + 2^{n+1}  + 2^{n+2}(\frac{n+2}{2})}{|\overline{C_{n+1}}| - n + 2^n(n-1) + 2^{n+2}(\frac{n+2}{2})} \\
    & = \frac{|\overline{C_n}| + 1+ n + 2^{n+1} + 2^{n+2}(\frac{n+2}{2})}{|\overline{C_{n+1}}| -n + 2^n(n-1) + 2^{n+2}(\frac{n+2}{2})}.
\end{align*}

For $j\geq 1$ such that $2^n(n-1) -n + 2^{n+2}(\frac{n+2}{2}) \leq |p_{n+2}|+j < |C_{n+2}|$ we use an automaton from Case $4$ and have that \[\frac{A(C[0..m])}{m+1} \leq \frac{|\overline{C_n}| +1 + n + 2^{n+1} + 2^{n+2}(\frac{n+2}{2})}{|\overline{C_{n+1}}| -n + 2^n(n-1) + 2^{n+2}(\frac{n+2}{2}) + j},\] i.e. the number of states required remains constant. Furthermore \[\lim_{n \to \infty}\frac{|\overline{C_n}| + 1 + n + 2^{n+1} + 2^{n+2}(\frac{n+2}{2}) }{|\overline{C_{n+1}}|-n + 2^n(n-1) + 2^{n+2}(\frac{n+2}{2})} =\frac{4}{7} < \frac{2}{3}.\]

\textbf{Case 2:} $n+1$ is a power of $2$

Suppose we are in Case $2$. Let $p_{n+2}$ be such that $C[0..m] = \overline{C_{n+1}}p_{n+2}$. We have a automaton as in Case $2$. It requires at most $|\overline{C_{n-1}}| + n + 2^n + 2^{n+1} + |p_{n+2}|$ states. Thus \begin{align*}
    \frac{A(C[0..m])}{m+1} & \leq \frac{|\overline{C_{n-1}}| + n + 2^n + 2^{n+1} + |p_{n+2}|}{|\overline{C_{n+1}}| + |p_{n+2}|}.
\end{align*}

For $|p_{n+2}|$ long enough, it is more beneficial to loop in $C_{n+2}$ instead of $C_n$ and use an automaton in the style of Case $1$ as $n+2$ is odd and $n+1$ is a power of $2$. As the final state will be contained in the second loop, such an automaton requires $|\overline{C_n}| + 2+ 2n  + 2^{n+1} + 2^{n+2}$ states.

So for $|p_{n+2}| \leq 2 + n + 2^n(n-1) + 2^{n+2}$ \begin{align*}
    \frac{A(C[0..m])}{m+1} &\leq \frac{|\overline{C_{n-1}}| + n + 2^n + 2^{n+1} + |p_{n+2}|}{|\overline{C_{n+1}}| + |p_{n+2}|} \\
    & \leq \frac{|\overline{C_{n-1}}| + 2 +n + n + 2^n + 2^n(n-1) + 2^{n+1}  + 2^{n+2}}{|\overline{C_{n+1}}| + 2 + n + 2^n(n-1) + 2^{n+2} } \\
    & = \frac{|\overline{C_{n}}| + 2 + 2n + 2^{n+1} + 2^{n+2}}{|\overline{C_{n+1}}| + 2 + n + 2^n(n-1) + 2^{n+2}}.
\end{align*}

For $j\geq 1$ such that $2+n + 2^n(n-1) + 2^{n+2} \leq |p_{n+2}|+j < |C_{n+2}|$ we use an automaton from Case $1$ and have that  \[\frac{A(C[0..m])}{m+1} \leq \frac{|\overline{C_n}| + 2+ 2n  + 2^{n+1} + 2^{n+2}}{|\overline{C_{n+1}}| + 2+n + 2^n(n-1) + 2^{n+2} + j},\] i.e. the number of states required remains constant. Furthermore \[\lim_{n \to \infty}\frac{|\overline{C_{n}}| + 2 + 2n + 2^{n+1} + 2^{n+2}}{|\overline{C_{n+1}}| + 2 + n + 2^n(n-1) + 2^{n+2}} = \frac{2}{5} < \frac{2}{3}.\]

\textbf{Case 4:} $n+1$ is even but not a power of $2$

Suppose we are in Case $4$, i.e. $n+1 = 2^st$ for some $s\geq 1$ and $t \geq 3$ odd. Let $p_{n+2}$ be such that $C[0..m] = \overline{C_{n+1}}p_{n+2}$. We have a automaton as in Case $4$. It requires at most $|\overline{C_{n-1}}|+ n + 2^n + 2^{n+1}t + |p_{n+2}|$ states.

Thus \begin{align*}
    \frac{A(C[0..m])}{m+1} & \leq \frac{|\overline{C_{n-1}}|+ n + 2^n + 2^{n+1}t + |p_{n+2}|}{|\overline{C_{n+1}}| + |p_{n+2}|}.
\end{align*}

For $|p_{n+2}|$ long enough, it is more beneficial to loop in $C_{n+2}$ instead of $C_n$ and use an automaton in the style of Case $3$ as $n+2$ is odd and $n+1$ is even but not power of $2$. As the final state will be contained in the second loop, such an automaton requires $|\overline{C_n}| + n  + 2^{n+1}t + 2^s + 2^{n+2}$ states. Hence, as $2^s \leq n/3$ and we have that the number of states required is bounded above by $|\overline{C_n}| + 4n/3  + 2^{n+1}t  + 2^{n+2}$.

So for $|p_{n+2}| \leq  n/3 + 2^n(n-1) + 2^{n+2}$ 
\begin{align*}
    \frac{A(C[0..m])}{m+1} & \leq \frac{|\overline{C_{n-1}}|+ n + 2^n + 2^{n+1}t + |p_{n+2}|}{|\overline{C_{n+1}}| + |p_{n+2}|} \\
    & \leq \frac{|\overline{C_{n-1}}|+ n + \frac{n}{3} + 2^n + 2^n(n-1) + 2^{n+1}t + 2^{n+2}}{|\overline{C_{n+1}}| + \frac{n}{3} + 2^n(n-1) + 2^{n+2}} \\
    & = \frac{|\overline{C_{n}}|+ \frac{4n}{3}   + 2^{n+1}t + 2^{n+2}}{|\overline{C_{n+1}}| + \frac{n}{3} + 2^n(n-1) + 2^{n+2}} \\
\end{align*}

For $j\geq 1$ such that $n/3 + 2^n(n-1) + 2^{n+2} \leq |p_{n+2}|+j < |C_{n+2}|$ we use an automaton from Case $3$ and have that \[\frac{A(C[0..m])}{m+1} \leq \frac{|\overline{C_{n}}|+ \frac{4n}{3}   + 2^{n+1}t + 2^{n+2}}{|\overline{C_{n+1}}| + \frac{n}{3} + 2^n(n-1) + 2^{n+2} + j},\] i.e. the number of states required remains constant. Furthermore \begin{align*}
    \lim_{n \to \infty} \frac{|\overline{C_{n}}|+ \frac{4n}{3}   + 2^{n+1}t + 2^{n+2}}{|\overline{C_{n+1}}| + \frac{n}{3} + 2^n(n-1) + 2^{n+2}} & \leq \lim_{n \to \infty} \frac{|\overline{C_{n}}|+ \frac{4n}{3}   + 2^{n+1}\frac{(n+1)}{2} + 2^{n+2}}{|\overline{C_{n+1}}| + \frac{n}{3} + 2^n(n-1) + 2^{n+2}} \\
    & = \frac{3}{5} < \frac{2}{3}.
\end{align*}

\subsection*{Analogous Result to Proposition \ref{loop lemma}}
The following is analogous for Proposition \ref{loop lemma} and demonstrates that if reading a zone $C_j$ where $j = 2^k$ for $k\geq 1$, any loop traversed used to read a substring $x$ of $C_j$, if $|x| \geq j$ then $|x|$ must be a multiple of $2^j-1$.

\begin{proposition}
Let $C\in \mathrm{PSC}$. Let $j = 2^k$ for some $k\geq 1$ and let $C'$ be a prefix of $C$ containing the substring $C_j$. Let $p_0 \rightarrow p_1 \rightarrow \cdots \rightarrow p_{2^jj}$ be the sequence of states an automaton which uniquely accepts $C'$ traverses while reading $C_j$. If there is some subsequence of states $p_i \rightarrow \cdots p_{i+l}  \rightarrow \cdots  p_{i + 2l}$ where $l \geq j$ and for all $0 \leq m \leq l$, $p_{i+m} = p_{i + l + m}$, then $l$ must be a multiple of $2^j-1$.
\label{loop lemma even}
\end{proposition}

\begin{proof}
Let $C$, $C'$, $j$ and $k$ be as above. Recall in such a case that \[C_j = d_j\cdot d_{j,1} \cdot d_{j,2} \cdots d_{j,2^{j-1}} = 0\cdot (d_j[1..2^j-1])^j\cdot 0^{j-1}\] for some de Bruijn string $d_j$.
First suppose a loop of length $l$ is traversed where $j \leq l \leq 2^j - 2$. Let $x$ be the substring of $C_j$ read during the loop. Hence $x = yz$ where $|y| = j$ and $|z| \leq 2^j - j - 2$. Suppose the loop is traversed twice in a row indicating that $x^2 = yzyz$ is a substring of $C_j$. Consider $yzy$ which has length at most $2^j + j -2$. By the construction of $C_j$, every substring of length $2^j + j - 2$ contains $2^j - 1$ of the strings of length $j$ as a substring exactly once where either $0^j$ or $10^{j-1}$ is the remaining string that does not appear. If $y = 0^j$ then $10^{j-1}$ is missing, otherwise $0^j$ is missing due to this shift between instances of the de Bruijn strings. However, $y$ is then a substring of length $j$ contained twice withing a substring of length $2^j + j - 2$ giving us a contradiction.

Next suppose a loop of length $l$ is traversed  where $d(2^j-1) < l < (d+1)(2^{j}-1)$ for some $d \leq \lfloor j/2 \rfloor$. Let $x$ be the string read while traversing the loop. Hence $x = yz$ where $|y| = d( 2^j-1)$ and $1 \leq |z| < 2^j-1$. Note that $y = (d_j[i..2^j-1]\cdot d_j[1..i-1])^d$ for some $i\geq 1$ ($i \neq 0$ as $\occ{0^j}{C_j} = 1$). Suppose the loop is traversed twice in a row indicating that $x^2 = yzyz$ is a substring of $C_j$. By the construction of $C_j$, this means that $z$ is a prefix of $d_j[i..2^j-1]\cdot d_j[1..i-1]$. This forces $z = \lambda$ or $z = d_j[i..2^j-1]\cdot d_j[1..i-1]$ which contradicts the length requirement of $z$.
\end{proof}

\end{document}